\documentclass{article} 
\usepackage{graphicx} % Required for inserting images

\usepackage[T1]{fontenc}
\usepackage[utf8]{inputenc} 
\usepackage{amssymb}
\usepackage{amsmath}
\usepackage{amsthm}
\usepackage{doi} 
\usepackage{braket}
\usepackage{hyperref}
\newtheorem{theorem}{Theorem}[]

\newtheorem{lemma}[theorem]{Lemma}
\newtheorem{remark}{Remark}[]

\DeclareMathOperator{\spec}{spec}

\usepackage{soul}
\usepackage{pgfplots}
\pgfplotsset{compat=1.18}

\usepackage{graphicx} %
\usepackage[left=25mm,right=25mm,top=25mm,bottom=25mm]{geometry}

\newcommand{\Tr}{\mathrm{Tr}}
\newcommand{\C}{\mathbb{C}}
\newcommand{\R}{\mathbb{R}}

\makeatletter
\let\saved@setaddresses\@setaddresses %
\let\@setaddresses\relax              %
\makeatother   

\usepackage[baseline]{euflag}
\usepackage[shortlabels]{enumitem}

\title{Convergence rates for the $k$-local Hamiltonian problem}
\title{Convergence rates for semidefinite hierarchies for $k$-local Hamiltonians}
\title{Quantitative semidefinite certificates for ground-state energies of Pauli Hamiltonians}
\author{Igor Klep\thanks{University of Ljubljana, Faculty of Mathematics and Physics, Jadranska 21, 1000 Ljubljana \&
University of Primorska, Faculty of Mathematics, Natural Sciences and Information Technologies,
Glagoljaška 8, 6000 Koper, Slovenia. Email: \texttt{igor.klep@fmf.uni-lj.si}} \and Nando Leijenhorst\thanks{Universit\'{e} de Toulouse; LAAS-CNRS, 7 avenue du colonel Roche, F-31400 Toulouse, France. Email: \texttt{nando.leijenhorst@laas.fr}} \and Victor Magron\thanks{Universit\'{e} de Toulouse; LAAS-CNRS, 7 avenue du colonel Roche, F-31400 Toulouse, France. Email: \texttt{victor.magron@laas.fr}}}
\date{\today}

\begin{document}

\maketitle % uncomment for arXiv version

\begin{abstract}
The $k$-local Hamiltonian problem is a central model for quantum many-body systems and Hamiltonian complexity.
Semidefinite programming and noncommutative sum-of-squares hierarchies provide systematic certificates for ground-state energies, but existing finite-convergence results give no quantitative guarantee on the accuracy of the low hierarchy levels accessible in computation. We prove explicit finite-level convergence rates for these hierarchies in the Pauli setting. 
For $k$-local Hamiltonians whose Pauli expansion contains only even-weight terms, we show that both the NPA-type lower-bound hierarchy and the upper-bound hierarchy on the spectral minimum have error at most
% \[
  $C(k)\xi^{n,4}_{d+1}/n$,
% \]
where $\xi^{n,4}_{d+1}$ is the smallest root of a Krawtchouk polynomial and $C(k)$ is independent of the number of qubits $n$ and the hierarchy level $d$. General $k$-local Hamiltonians reduce to this even-weight case by adding one ancilla qubit while preserving the spectrum. The proof constructs almost-reproducing kernels for the Pauli algebra and relates their spectra to Krawtchouk polynomials, giving a noncommutative analogue of recent kernel-based convergence analyses for commutative polynomial optimization. 
These results provide the first quantitative finite-level accuracy guarantees for noncommutative semidefinite relaxations of Pauli Hamiltonians.
\end{abstract}

\section{Introduction}

Certifying the ground-state energy of a many-body quantum system is a basic task in quantum information \cite{kempe_complexity_2006,GHLS15}, condensed-matter physics \cite{BF04,WSFLRMA24}, and quantum chemistry \cite{WBAG11,Mazziotti}. For a system of $n$ qubits with local interactions, this task is captured by the $k$-local Hamiltonian problem: given a Hamiltonian
\begin{equation}\label{eq:k-local_hamiltonian}
H = \sum_{\substack{S \subseteq \{1, \dots, n\} \\ |S| \leq k}} H_S,
\end{equation}
where $H_S$ acts nontrivially only on the qubits indexed by $S$,
estimate or certify the smallest eigenvalue of $H$.
Since the $k$-local Hamiltonian problem is QMA-hard already for fixed locality $k\geq 2$~\cite{kempe_complexity_2006,GHLS15}, rigorous approximation and relaxation methods are indispensable.

A large body of work therefore studies tractable approximation methods for local Hamiltonians and important special cases such as Quantum Max Cut, a $2$-local problem without $1$-local terms that is itself QMA-hard \cite{piddock_complexity_2017}. Variational methods, including restrictions to product states~\cite{gharibian_approximation_2012,BrandaoHarrow2016} or more general Ans\"atze~\cite{king_improved_2023}, produce feasible states and hence upper bounds on the smallest eigenvalue. In this paper, we focus on the complementary certification problem: obtaining rigorous lower bounds through semidefinite programming and noncommutative sum-of-squares relaxations.

Semidefinite programming relaxations and noncommutative sum-of-squares methods~\cite{HM04,NPAog,NPA2008,burgdorf_optimization_2016} produce certified lower bounds on the ground-state energy through a hierarchy of tractable relaxations. 
These hierarchies are complete in the sense that they converge at sufficiently high levels. However, completeness alone does not answer the computationally relevant question: how accurate is a fixed, low-level relaxation?
Without an explicit rate, a hierarchy may be theoretically convergent but provide no certified information at the levels that can actually be solved.

In this paper, we give the first explicit finite-level accuracy guarantees for sum-of-squares hierarchies for noncommutative polynomial optimization problems in the Pauli setting. More precisely, we analyze both the NPA hierarchy \cite{PNA2010} and the upper-bound hierarchy for the spectral minimum \cite{klep2025upperboundhierarchiesnoncommutative}. For $k$-local Hamiltonians whose Pauli expansion contains only even-weight terms, the $d$-th level of both hierarchies have error bounded by\looseness=-1
\[
  C(k)\frac{\xi^{n,4}_{d+1}}{n},
\]
where $\xi^{n,4}_{d+1}$ is the smallest root of a Krawtchouk polynomial and $C(k)$ depends only on the locality $k$. For the $k$-local Hamiltonian problem, the even-weight assumption corresponds to Hamiltonians of the form~\eqref{eq:k-local_hamiltonian} with $H_S=0$ whenever $|S|$ is odd. This restriction is not essential for spectral optimization: an arbitrary $k$-local Hamiltonian on $n$ qubits can be embedded into an even-weight Hamiltonian on $n+1$ qubits with the same spectrum.

{
The polynomial kernel method has been recently leveraged to obtain convergence rates for commutative polynomial optimization on the sphere~\cite{fang_sum--squares_2021}, the hypercube~\cite{laurent_effective_2023}, the binary cube~\cite{slot_sum--squares_2023}, the ball, and the simplex~\cite{slot_sum--squares_2024}. 
For a comprehensive overview of convergence-rate results in commutative polynomial optimization, we refer to the recent survey \cite{OverviewConvergenceRates}. 
Our proof adapts the kernel method used in \cite{fang_sum--squares_2021,slot_sum--squares_2023,slot_sum--squares_2024} to the noncommutative setting of Pauli Hamiltonians. 
While our overall strategy is similar to that of the three works above, our convergence rate estimate connects closely to the estimates for the binary cube in~\cite{slot_sum--squares_2023}, which can be seen as a special case of our setting. In particular, both convergence rate estimates depend heavily on Krawtchouk polynomials (although with different parameters).}
We construct kernels on the Pauli algebra that almost reproduce low-degree even Pauli polynomials. After diagonalization, the spectra of these kernels are governed by Krawtchouk polynomials with parameter $q=4$, which leads to the explicit rate above.

\subsection{Main results}
We now state the main convergence estimates. We first recall the
noncommutative polynomial formulation of a Pauli Hamiltonian.

Let $\C\langle x\rangle/\mathcal I$ be the quotient of the free
$*$-algebra generated by variables $x_{\sigma,i}$, where
$\sigma\in\{\sigma_X,\sigma_Y,\sigma_Z\}$ and $i\in\{1,\ldots,n\}$, by the ideal $\mathcal I$ generated by the relations satisfied by the corresponding Pauli operators. 
Evaluation at the
Pauli matrices will be denoted by $p\mapsto p(\sigma)$.
Since Pauli words form a basis of the matrix algebra on $n$ qubits, every Hamiltonian $H$ has a unique representative $p_H\in\mathbb{C}\langle x\rangle/\mathcal I$ such that $H=p_H(\sigma)$. 
If $H$ is $k$-local, then $p_H$ has degree at most $k$. If $H_S=0$ whenever $|S|$ is odd, then $p_H$ contains only even-degree monomials, and we call such a polynomial an even-weight polynomial. {In the remainder of the paper, we work with a general polynomial $p$, since $H$ is uniquely determined by $p_H$.}
Here the degree, or weight, of a Pauli word is the number of qubits on which it acts nontrivially.
For more background on noncommutative polynomial optimization, see \cite{burgdorf_optimization_2016}.\looseness=-1

Let $\Sigma_d\langle x\rangle$ denote the cone of sums-of-Hermitian-squares polynomials of degree at most $2d$. The hierarchy of lower bounds on the spectral minimum can then be written as
\begin{equation}
    \label{eq:main_hierarchy}
\begin{aligned}
    \nu_d(p) & = \sup && \lambda, \\
    & \phantom{{}=}\text{ s.t.} && p- \lambda \in \Sigma_{d}\langle x \rangle/\mathcal I.
\end{aligned}
\end{equation}
Since every feasible $\lambda$ satisfies $\lambda\le \lambda_{\min}(p)$, the numbers $\nu_d(p)$ are certified lower bounds on the spectral minimum.
The hierarchy \eqref{eq:main_hierarchy} is known to converge at level $n$: we have $p - \lambda_{\min}(p(\sigma)) = s^*s$, where $s(\sigma) = G$ for some $G$ with $G^*G = p(\sigma) -\lambda_{\min}(p(\sigma))I_{2^n}$. The polynomial $s$ is of degree at most $n$ since the Pauli words of degree at most $n$ form a basis of the matrix algebra, so this gives a solution to the $n$-th level of the hierarchy. 
In Remark~\ref{rmk:finite_conv}, we recover this finite convergence through our proof. 
\begin{theorem}\label{thm:main}
    Let $p \in \C\langle x \rangle_{k}/\mathcal I$ be a Hermitian even-weight polynomial. 
Let $\lambda_{\min}(p)$ denote the smallest eigenvalue of $p(\sigma)$, and suppose that $\|p(\sigma)\|_{\infty} \leq 1$, where $\|\cdot\|_\infty$ is the spectral norm.
Let $\xi_d^{n,4}$ be the smallest root of the Krawtchouk polynomial of degree $d$ with parameters $n$ and $q=4$. If $\frac{4k}{3} \frac{\xi_{d+1}^{n,4}}{n}\leq \frac{1}{2}$,  
then 
\begin{equation}\label{eq:main_rate}
  \lambda_{\min}(p)-\nu_d(p) \leq C(k)\frac{\xi_{d+1}^{n,4}}{n},
    \end{equation}
    where $C(k)$ is independent of $n,d$ and satisfies
    \[
    C(k)< \frac{2}{3}k(k+2)(1+\sqrt{2})^{2k+1}.
    \]
\end{theorem}

The even-degree hypothesis is a technical assumption 
of the proof, not a restriction on the spectra that can be treated.
In Appendix~\ref{app:evenweight}
we show that every $k$-local Hamiltonian on $n$ qubits can be embedded into an
even-weight Hamiltonian on $n+1$ qubits with the same spectrum, with locality
$k'=k$ for even $k$ and $k'=k+1$ for odd $k$. This extends
\cite[Lemma~1]{bravyi_approximation_2019}, which treats the case $k=2$. 

In the proof of Theorem~\ref{thm:main} (Section \ref{sec:proof_thm1}), we choose a parameter $\eta = 2$ for simplicity. However, for a given degree $d$ the optimal $\eta > 0$ equals
\[
\eta = \Big(1-\frac{4k}{3}\frac{\xi_{d+1}^{n,4}}{n}\Big)^{-1},
\]
which improves the necessary condition on the degree in Theorem \ref{thm:main} to
\[
\frac{4k}{3} \frac{\xi_{d+1}^{n,4}}{n} < 1.
\]
In this case, the bound in \eqref{eq:main_rate} changes by a factor $\eta/2$.
See Section~\ref{sec:bound_sum_c2r} for more details.

The hierarchy \eqref{eq:main_hierarchy} is the noncommutative analog of the moment/sum-of-squares hierarchy in (commutative) polynomial optimization \cite{lasserre_global_2001}, which gives lower bounds on a minimum. Similarly, there is a noncommutative analog to Lasserre's hierarchy of upper bounds on a minimum of a polynomial \cite{lasserre_upperbound_2011}, which gives upper bounds on the spectral minimum \cite{klep2025upperboundhierarchiesnoncommutative}. We can use the functions we construct for the proof of Theorem~\ref{thm:main} to give a convergence rate for this hierarchy as well, similarly to the commutative case \cite{slot_sum--squares_2023, slot_sum--squares_2024,deKlerkLaurentSun2015}.

The (dual of the) upper bound hierarchy is given by
\begin{equation}\label{eq:second_hierarchy}
    \begin{aligned}
        \mu_d(p) &=  \inf && \langle p, s\rangle, \\
        &\phantom{{}=} \text{ s.t.} && \langle 1, s\rangle = 1, \\
        & && s \in \Sigma_d\langle x \rangle/\mathcal I,
    \end{aligned}
\end{equation}
where in our case $\langle \cdot, \cdot\rangle$ denotes the normalized trace inner product
induced by evaluation on Pauli matrices. 
{With $s$ such that $s(\sigma) = vv^\dagger$, where 
$v$ is an eigenvector of $p(\sigma)$ corresponding to the minimum eigenvalue, 
we can use a similar argument as for the hierarchy \eqref{eq:main_hierarchy} to show that the hierarchy \eqref{eq:second_hierarchy} converges at level $n$.}

\begin{theorem}\label{thm:second}
     Let $p \in \C\langle x \rangle_{k}/\mathcal I$ be a Hermitian even-weight polynomial. 
     Let $\lambda_{\min}(p)$ denote the smallest eigenvalue of $p(\sigma)$, and suppose that $\|p(\sigma)\|_{\infty} \leq 1$.
Let $\xi_d^{n,4}$ be the smallest root of the Krawtchouk polynomial of degree $d$ with parameters $n$ and $q=4$. Then
\begin{equation}\label{eq:second_rate}
    \mu_d(p)-\lambda_{\min}(p) \leq \frac{C(k)}{2}\frac{\xi_{d+1}^{n,4}}{n},
    \end{equation}
    where $C(k)$ is the same constant as in Theorem~\ref{thm:main}.
\end{theorem}

Unlike Theorem \ref{thm:main},
the upper-bound estimate \eqref{eq:second_rate} holds for every relaxation level
$d$.

\begin{figure}[t]
\centering
\begin{tikzpicture}
\begin{axis}[
    width=0.92\textwidth,
    height=0.62\textwidth,
    xmin=0, xmax=0.75,
    ymin=0, ymax=0.8,
    xlabel={Relaxation fraction $d/n$},
    ylabel={Normalized smallest root},
    xtick={0,0.05,0.10,0.15,0.20,0.25,0.30,0.35,0.40,0.45,0.50,0.55,0.60,0.65,0.70,0.75},
    xticklabel style={
        /pgf/number format/fixed,
        /pgf/number format/precision=2,
        /pgf/number format/fixed zerofill
    },
    grid=major,
    grid style={gray!25},
    tick label style={font=\small},
    label style={font=\small},
    legend style={
        at={(0.98,0.98)},
        anchor=north east,
        draw=gray!50,
        fill=white,
        font=\small
    },
]
% Exact normalized smallest root xi_{d+1}^{40,4}/40
\addplot[
    blue,
    thick,
    mark=*,
    mark size=1.4pt
] coordinates {
    (0.000, 0.750000)
    (0.025, 0.675000)
    (0.050, 0.617302)
    (0.075, 0.568446)
    (0.100, 0.525296)
    (0.125, 0.486280)
    (0.150, 0.450478)
    (0.175, 0.417297)
    (0.200, 0.386327)
    (0.225, 0.357273)
    (0.250, 0.329916)
    (0.275, 0.304086)
    (0.300, 0.279650)
    (0.325, 0.256504)
    (0.350, 0.234563)
    (0.375, 0.213758)
    (0.400, 0.194035)
    (0.425, 0.175348)
    (0.450, 0.157661)
    (0.475, 0.140946)
    (0.500, 0.125181)
    (0.525, 0.110350)
    (0.550, 0.096443)
    (0.575, 0.083453)
    (0.600, 0.071380)
    (0.625, 0.060227)
    (0.650, 0.050000)
    (0.675, 0.040710)
    (0.700, 0.032372)
    (0.725, 0.025000)
    (0.750, 0.018611)
};
\addlegendentry{$\xi_{d+1}^{40,4}/40$}

% Asymptotic curve phi_4(t)
\addplot[
    orange,
    thick,
    dashed,
    domain=0:0.75,
    samples=200
]
{0.75 - (0.5*x + 0.5*sqrt(3*x*(1-x)))};
\addlegendentry{$\phi_4(t)$}

% Theorem condition lines:
% xi/n <= 3/(4k)

\addplot[
    black!65,
    dashed,
    thick,
    domain=0:0.75,
    samples=2
]
{3/(4*2)};
\addlegendentry{$3/(4k)$, $k=2$}

% \addplot[
%     black!70,
%     dotted,
%     thick,
%     domain=0:0.75,
%     samples=2
% ]
% {3/(4*6)};
% \addlegendentry{$3/(4k)$, $k=6$}

\addplot[
    black!65,
    dotted,
    thick,
    domain=0:0.75,
    samples=2
]
{3/(4*4)};
\addlegendentry{$3/(4k)$, $k=4$}

\end{axis}
\end{tikzpicture}
\caption{
Normalized smallest Krawtchouk roots controlling the convergence rates in
Theorems~\ref{thm:main} and~\ref{thm:second}. The solid curve shows
$\xi_{d+1}^{40,4}/40$, while the dashed orange curve is the asymptotic limit
$\phi_4(t)=\frac34-\left(\frac{t}{2}+\frac12\sqrt{3t(1-t)}\right)$.
The two horizontal lines show the optimized sufficient condition
$\xi_{d+1}^{n,4}/n\le 3/(4k)$ for Theorem~\ref{thm:main} (as explained in the paragraph following the theorem), for $k=2,4$.}
\label{fig:krawtchouk-rate-curve}
\end{figure}
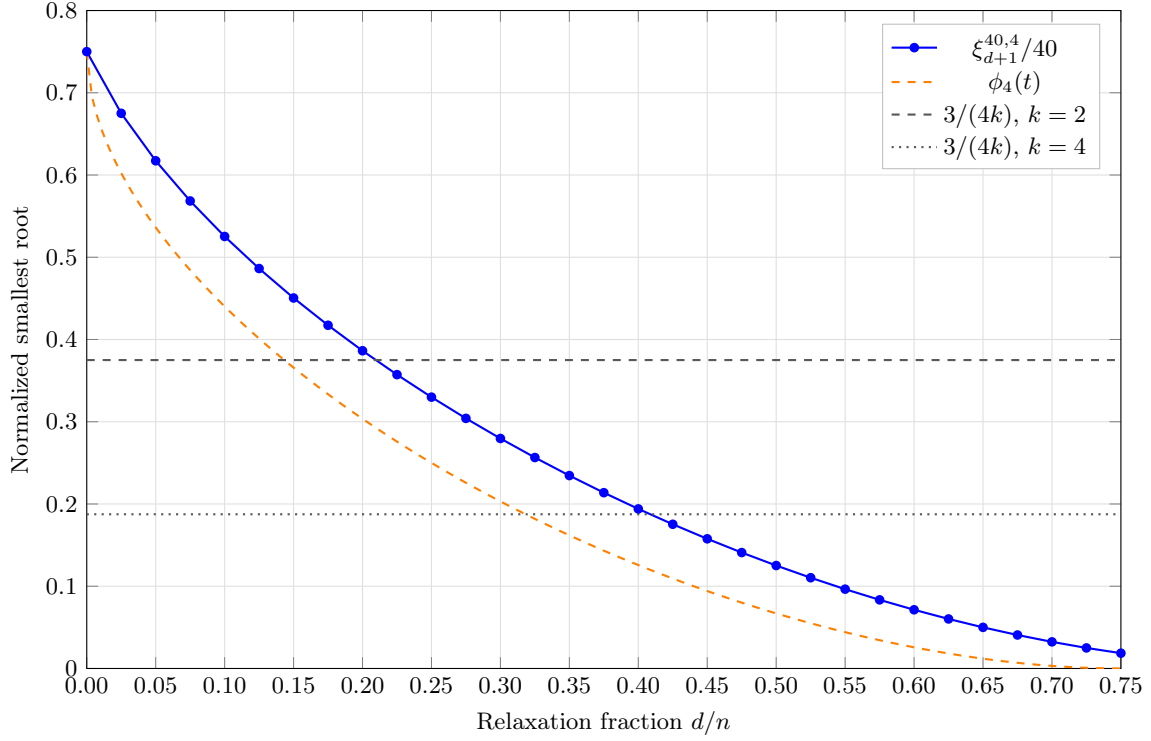

The Krawtchouk-root expression in Theorems~\ref{thm:main} and~\ref{thm:second} can be made more explicit
using known estimates for the smallest roots.
By \cite[Corollary~5.20]{vi_universal_1998}, the smallest roots of the Krawtchouk polynomials with parameter $q>1$ 
satisfy 
\[
\frac{\xi_{d}^{n,q}}{n} \leq \frac{q-1}{q} - \frac{q-2}{2q}\frac{h_d^2}{n} - \frac{1}{q}\sqrt{2(q-1)\left(1-\frac{d+2}{n}\right)} \frac{h_d}{\sqrt{n}}
\]
where $h_d$ is the largest root of the Hermite polynomial of degree $d$, which satisfies
\[
h_d = \sqrt{2d} - c(2d)^{-1/6} + o(d^{-1/6})
\]
as $d \to \infty$, where $c = 1.85575\dots$ is a constant.
Asymptotically, we have by \cite[Corollary~5.21]{vi_universal_1998} (see also \cite[Theorem~34]{slot_sum--squares_2023}) that for $t \in [0, \frac{q-1}{q}]$,
\[
\lim_{d/n\to t} \frac{\xi_d^{n,q}}{n} = \phi_q(t) = \frac{q-1}{q} - \left( \frac{q-2}{q}t + \frac{2}{q}\sqrt{(q-1)t(1-t)}\right),  
\]
where the limit means that for any sequences $(n_j)_j, (d_j)_j$ such that $\lim_{j \to \infty} n_j = \infty$ and $\lim_{j \to \infty} d_j/n_j = t$, we have $\lim_j \xi_{d_j}^{n_j, q}/n = \phi_q(t)$.

\subsection{Proof technique}
The proof adapts the Christoffel--Darboux kernel method used to obtain
convergence rates for commutative sum-of-squares hierarchies
\cite{fang_sum--squares_2021, slot_sum--squares_2023, slot_sum--squares_2024}. The main new ingredient is a noncommutative analogue of this construction for the Pauli algebra.

We use reproducing kernels to find an invertible linear map $\mathbf{K}: \C\langle x \rangle_{2d}/\mathcal I \to \C\langle x \rangle_{2d}/ \mathcal I$ (with $2d\geq k$) satisfying the following properties:
\begin{align}
    &\mathbf{K}1 = 1 &&\tag{P1} \label{eq:property1} \\
    &\mathbf{K}p \in \Sigma_d\langle x \rangle/\mathcal I && \text{if $p(\sigma)$ is positive semidefinite} \tag{P2}\label{eq:property2} \\
    &\| \mathbf{K}^{-1}p - p\|_{\infty} \leq \varepsilon && 
    \text{if $p$ is an even-weight polynomial of degree $k$ with $\|p\|_{\infty} \leq 1$}
        \tag{P3}\label{eq:property3}
\end{align}
where $\|p\|_\infty$ is the spectral norm after evaluation on the Pauli matrices.

\begin{lemma}\label{lem:properties_K}
    Let $\mathbf K$ be an invertible linear map satisfying properties \eqref{eq:property1}-\eqref{eq:property3}. Let $p \in \C\langle x \rangle_k /\mathcal I$ be Hermitian and suppose that every monomial of $p$ has even degree with $\|p\|_{\infty} \leq 1$. Then $\lambda_{\min}(p)-\nu_d(p) \leq \varepsilon$. 
\end{lemma}
\begin{proof}
Since $\mathbf K1=1$, we have
\[
\mathbf K^{-1}(p-\lambda_{\min}(p) +\varepsilon)=\mathbf K^{-1}p-\lambda_{\min}(p)+\varepsilon.
\]
By \eqref{eq:property3}, the polynomial $\mathbf K^{-1}p$ differs from $p$ in spectral norm by
at most $\varepsilon$. Since $p(\sigma)\succeq \lambda_{\min}(p) {I_{2^n}}$, it follows that
\[
\mathbf K^{-1}(p-\lambda_{\min}(p)+\varepsilon)(\sigma)\succeq 0.
\]
Applying \eqref{eq:property2} to $\mathbf K^{-1}(p-\lambda_{\min}(p)+\varepsilon)$ gives
\[
p-\lambda_{\min}(p)+\varepsilon
=
\mathbf K\big(\mathbf K^{-1}(p-\lambda_{\min}(p)+\varepsilon)\big)
\in \Sigma_d\langle x\rangle/\mathcal I.
\]
Thus $\nu_d(p)\geq \lambda_{\min}(p)-\varepsilon$.
\end{proof}

Thus the problem reduces to constructing a positive kernel that is close to the identity on low-degree even Pauli polynomials. We start from the reproducing kernel associated with the normalized trace inner product and perturb its homogeneous components so that the resulting kernel is positive while remaining almost reproducing.

\subsection{Guide to the paper}

Section~\ref{sec:preliminaries} recalls the kernel formalism and the basic facts about Krawtchouk polynomials used throughout the proof. In Section~\ref{sec:repr_kernels} we construct the homogeneous reproducing kernels for the Pauli algebra and diagonalize them, showing that their spectra are governed by Krawtchouk polynomials with parameter $q=4$. Section~\ref{sec:proof_thm1} uses this diagonalization to construct an almost-reproducing positive kernel and prove the lower-bound estimate of Theorem~\ref{thm:main}. Section~\ref{sec:proof_thm2} applies the same kernel to the dual upper-bound hierarchy and proves Theorem~\ref{thm:second}. Section~\ref{sec:concl} discusses limitations and possible extensions, including qudit systems, swap-operator formulations, and sparse hierarchies. Appendix~\ref{app:evenweight} gives the reduction from general local Hamiltonians to even-weight Hamiltonians, while Appendices~\ref{app:proofs_lemmas} and \ref{app:bound_gamma_k} contain the auxiliary estimates on Krawtchouk polynomials and homogeneous Pauli components.

\section{Preliminaries}\label{sec:preliminaries}
\subsection{Kernels}
\label{sec:prelim_kernels}
Let $P \subseteq \C\langle x\rangle$ be a polynomial space with inner product $\langle \cdot, \cdot\rangle$. A kernel is an element of $P \otimes P$, and there is a natural inner product on $P \otimes P$ given by $\langle a_1 \otimes b_1, a_2 \otimes b_2\rangle = \langle a_1, a_2\rangle\langle b_1, b_2\rangle$. 
{We use the convention that inner products are linear in the second argument.}
For $K=\sum_{i=1}^N b_{1,i} \otimes b_{2,i}^* \in P \otimes P$, where $*$ denotes the involution of $\C\langle x\rangle$ 
 (representing the Hermitian conjugate), we define the operator $\mathbf K:P \to P$ by
\begin{equation}\label{eq:fatK}
\mathbf Kp = \sum_{i=1}^N b_{1,i} \langle b_{2,i}, p\rangle.
\end{equation}
{We often evaluate the second factor of the kernel $K$ on Pauli matrices, while the first factor remains a polynomial in $x$, to obtain $K(x, \sigma) = \sum_{i=1}^N b_{1,i}(x) \otimes b_{2,i}(\sigma)^*$. 
We will use $\langle p_1, p_2\rangle = \Tr(p_1(\sigma)^*p_2(\sigma))$,
where $\Tr$ denotes the normalized trace, and where we denote by $p(\sigma)$ the evaluation of $p$ on the matrices
\begin{align}
\label{def:pauli}
\sigma^i = I_2 \otimes \dots \otimes \sigma \otimes \dots \otimes I_2.
\end{align}
Here the matrix $\sigma \in \{\sigma_X, \sigma_Y, \sigma_Z\}$ acts on the $i$-th qubit. 
Then we obtain the following expression equivalent to \eqref{eq:fatK}:}
\[
\mathbf Kp = {\sum_{i=1}^N b_{1,i}(x) \Tr(b_{2,i}(\sigma)^* p(\sigma) )} = \Tr_2(K(x, \sigma)({I_{2^n}} \otimes p(\sigma))),
\]
where $\Tr_2$ is the partial trace defined by $\Tr_2(A \otimes B) = \Tr(B) A $ and extended linearly.

Now if $b_{1,i} = b_{2,i} = b_i$ and $\{b_i\}$ is an orthonormal basis with respect to the inner product, then $K$ is a reproducing kernel: the associated linear operator satisfies $\mathbf Kp = p$. The reproducing kernel depends only on the inner product, and not on the choice of orthonormal basis. Such a kernel is also sometimes called a Christoffel-Darboux kernel.

\subsection{Krawtchouk polynomials}
Let $q \geq 2$ and let $w_q:\{0, \dots, n\} \to \R$ be the weight function
\[
w_q(i) = q^{-n}(q-1)^i{n \choose i}.
\]
This defines an inner product 
\[
\langle p_1, p_2\rangle_{q} = \sum_{i=0}^n p_1(i)p_2(i)w_q(i),
\]
where $p_1$ and $p_2$ are univariate polynomials with real coefficients.
The standard orthogonal polynomials with respect to this inner product are the Krawtchouk polynomials 
\[
\mathcal K_r^{n,q}(i) = \sum_{j=0}^r(-q)^j(q-1)^{r-j}{n-j \choose r-j}{i \choose j}
\]
of degree $r$ (see, e.g., \cite[Section~5.7]{macwilliams_theory_1983}).
Note that this is not the standard definition of the Krawtchouk polynomials, but an equivalent expression that will come up more naturally in this paper.

The Krawtchouk polynomials satisfy the orthogonality relations
\[
\langle \mathcal K_r^{n,q}, \mathcal K_{r'}^{n,q}\rangle_{q} = \delta_{r,r'}(q-1)^r{n \choose r}.
\]
Furthermore, $\mathcal K_r^{n,q}(0) = (q-1)^r{n \choose r} = \|\mathcal K_r^{n,q}\|^2_{q}$, where $\| \cdot \|_q$ is the norm induced by the above inner product. 
We also consider the dual normalization
\[
\widehat{\mathcal K}_r^{n,q}(i) = \frac{\mathcal K_r^{n,q}(i)}{\|\mathcal K_r^{n,q}\|_q^2}
\]
so that 
\[
\langle \widehat{\mathcal K}_r^{n,q}, \mathcal K_{r'}^{n,q}\rangle_{q} = \delta_{r,r'}.
\]

The following two elementary lemmas will be used to control the coefficients of the almost-reproducing kernel.

\begin{lemma}\label{lem:nk1}
For every $0 \le i \le n$ and $0 \le r \le n$,
\[
\mathcal K_r^{n,q}(i) \le \mathcal K_r^{n,q}(0)=(q-1)^r\binom{n}{r}.
\]
\end{lemma}

\begin{lemma}\label{lem:aff_bound_k}
For every $0 \le i \le n$ and $1 \le r \le n$,
\[
\mathcal K_r^{n,q}(0)-\mathcal K_r^{n,q}(i)
\le
q \cdot (q-1)^{r-1}\binom{n-1}{r-1}\, i.
\]
Equivalently, 
\[
1-\widehat{\mathcal K}_r^{n,q}(i)\le \frac{qr}{(q-1)n}\,i.
\]
\end{lemma}
We prove both lemmas in Appendix~\ref{app:proofs_lemmas}. Lemma~\ref{lem:aff_bound_k} is a slight strengthening of Lemma~33 in \cite{slot_sum--squares_2023} for $q > 2$, which states that
\[
1 - \widehat{\mathcal K}_r^{n,q}(i) \leq \frac{2r}{n}i.
\]

\section{Reproducing kernels for the Pauli algebra}\label{sec:repr_kernels}
Recall that we consider polynomials in variables $x_{\sigma,i}$ where $\sigma$ is a Pauli matrix and $i \in \{1, \dots, n\}$. Furthermore, the polynomials have no terms with multiple variables corresponding to the same $i$, and the variables $x_{\sigma, i}$ satisfy the same relations as the Pauli matrices. 
The inner product on this polynomial space is given by
\[
\langle p_1, p_2\rangle = \Tr(p_1^*(\sigma)p_2(\sigma)).
\]
Recall that the Pauli matrices together with the identity form an orthonormal basis with respect to the trace inner product. {Thus the tensor products of these matrices form an orthonormal basis of $M_{2^n \times 2^n}(\C)$.}
In particular, this means that the words 
\[
\prod_{j=1}^r x_{\sigma_j,i_j}
\]
form a basis of the space of polynomials of degree $r$, {that is orthonormal with respect to the trace inner product}, where $\sigma_j \in \{\sigma_X, \sigma_Y, \sigma_Z\}$ and $1\leq i_j< i_{j+1} \leq n$ for all $j$.

Define the reproducing kernels for the spaces of homogeneous polynomials of degree $r$ by 
\begin{align}
\label{eq:defCrn}
C_r^{n} = \sum_{1 \leq i_1 < \dots < i_r\leq n} \sum_{\sigma_j \in \{\sigma_X, \sigma_Y, \sigma_Z\}} \prod_{j=1}^rx_{\sigma_j, i_j} \otimes x_{\sigma_j, i_j}.
\end{align}
Then the reproducing kernel for the space of polynomials up to degree $d$ is given by
\[
\sum_{r=0}^d C_r^{n}.
\]

Now evaluation of $C_r^{n}$ on the Pauli matrices gives
\[
C_r^{n}(\sigma,\sigma) = \sum_{\substack{J \subseteq [n] \\ |J| = r}}\prod_{j \in J}\sum_{W \in \{X, Y, Z\}} \sigma_W^{j} \otimes \sigma_W^j. \]
In the following sections, it will be useful to write $C_r^{n}(\sigma,\sigma)$ in terms of the matrices 
\[A_j = \frac{3}{4}{I_{4^n}} + \frac{1}{4}\sum_{W \in \{X, Y, Z\}} \sigma_W^j \otimes \sigma_W^j.
\] 
Thus $A_j$ acts on the two-copy space of the $j$-th qubit. 
This gives
\[
C_r^{n}(\sigma,\sigma) = \sum_{\substack{J \subseteq [n] \\ |J| = r}} \prod_{j \in J} (4A_j - 3{I_{4^n}}) = \sum_{l=0}^r (4)^l(-3)^{r-l}{n-l \choose r-l} \sum_{\substack{J \subseteq [n] \\ |J| = l}}\prod_{j \in J} A_j,
\]
{as a consequence of pairwise commutativity of the $A_j$'s.}

\subsection{Reduction to the Krawtchouk polynomials}
We now diagonalize the  kernels $C_r^n(\sigma,\sigma)$ by expressing the eigenvalues of $C_r^n(\sigma, \sigma)$ in terms of Krawtchouk polynomials with parameter $q=4$. 

Let $A = \frac{3}{4} {I_4} + \frac{1}{4}\sum_{W \in \{X, Y, Z\}} \sigma_W \otimes \sigma_W$ be the single-site version of the operators $A_j$. It has eigenvalues $0$ and $1$ with multiplicity $1$ and $3$, respectively. 
Let $Q$ be an orthogonal matrix such that
\[
Q^{\sf T}AQ= \mathrm{Diag}(0, 1, 1, 1).
\]
Then $Q^{\otimes n}$ diagonalizes the matrices $C_r^{n}(\sigma,\sigma)$, where the
$j$-th copy of $Q$ acts on the two-copy space corresponding to the $j$-th qubit. 
The diagonal entries of $(Q^{\otimes n})^{\sf T} C_r^n(\sigma,\sigma) Q^{\otimes n}$ are then given by
\begin{align}
((Q^{\otimes n})^{\sf T} C_r^{n}(\sigma,\sigma) Q^{\otimes n})_{a,a} &= \sum_{l=0}^r 4^l(-3)^{r-l}{n-l \choose r-l} \sum_{\substack{J \subseteq [n] \\ |J| = l}}\prod_{j \in J} a_j \nonumber \\
&= (-1)^{r}\sum_{l=0}^r (-4)^{l}(3)^{r-l}{n-l \choose r-l} {\sum_j a_j \choose l} \label{eq:krawappears}
\end{align}
where $a \in \{0,1\}^n$ indicates that we consider the eigenvalues $a_j$ of $A_j$. In particular, the expression \eqref{eq:krawappears} is $(-1)^r \mathcal K_r^{n,q}(i)$ with $i = \sum_j a_j \in \{0, \dots, n\}$ and $q = 4$.

We record the conclusion as a lemma.

\begin{lemma}
\label{lem:eigCrn}
For $0\le r\le n$, the eigenvalues of $C_r^n(\sigma,\sigma)$ are
$(-1)^r\mathcal K_r^{n,4}(i)$, $i=0,\ldots,n.$ 
The eigenvalue corresponding to $i$ has multiplicity
$
3^i{n\choose i}.$
Consequently, the kernels $C_r^n$ are pairwise orthogonal with respect to
the trace inner product, and
\[
\|C_r^n\|^2=3^r{n\choose r}.
\]
\end{lemma}

\begin{proof}
The diagonalization above gives the eigenvalue formula. It remains only to
account for the multiplicities and the norm. The eigenspace of a single
$A_j$ with eigenvalue $1$ has dimension $3$, while the eigenspace with
eigenvalue $0$ has dimension $1$. Hence the number of joint eigenvectors
with $\sum_j a_j=i$ is
\[
3^i{n\choose i}.
\]
Since we use the normalized trace on the $4^n$-dimensional 2-copy space,
the induced weight on $i\in\{0,\ldots,n\}$ is\looseness=-1
\[
w(i)=4^{-n}3^i{n\choose i}.
\]
This is exactly the Krawtchouk weight $w_4(i)$. Therefore
\[
\langle C_r^n,C_s^n\rangle
=
\sum_{i=0}^n
(-1)^r\mathcal K_r^{n,4}(i)
(-1)^s\mathcal K_s^{n,4}(i)
w_4(i)
= 
{ (-1)^{r+s}
\langle \mathcal K_r^{n,4}, \mathcal K_s^{n,4}
\rangle_4}
.
\]
By the orthogonality relations for the Krawtchouk polynomials, this equals
$0$ for $r\neq s$, and for $r=s$ it equals
\[
3^r{n\choose r}.\qedhere
\]
\end{proof}

Finally, the same identity can be written symbolically in the polynomial
algebra. Since the variables $x_{\sigma,j}$ satisfy the Pauli relations,
we have
\begin{equation}
\label{eq:K_to_krawtchouk}
C_r^n(x,x')
\equiv
(-1)^r
\mathcal K_r^{n,4}\Big(\sum_{j=1}^n A_j(x,x')\Big)
\pmod{\mathcal I\otimes\mathcal I},
\end{equation}
where
\[
A_j(x,x')
=
\frac34\,1\otimes 1
+
\frac14\sum_{\sigma\in\{\sigma_X,\sigma_Y,\sigma_Z\}}
x_{\sigma,j}\otimes x'_{\sigma,j}.
\]
{
The equality stated in \eqref{eq:K_to_krawtchouk} holds in the quotient algebra 
$\pmod{\mathcal I\otimes\mathcal I}$ obtained after imposing the Pauli relations separately in each tensor factor.
}

\section{Proof of Theorem~\ref{thm:main}}\label{sec:proof_thm1}

Recall that by Lemma~\ref{lem:properties_K}, it suffices to find a linear operator $\mathbf K$ with the properties
\begin{align}
&\mathbf K1 = 1, \tag{P1} \\
&\mathbf Kp \in \Sigma_{d}\langle x \rangle/\mathcal I && \text{ if } p(\sigma) \succeq 0, \tag{P2}\\
 &\|\mathbf K^{-1}p - p \|_{\infty} \leq \varepsilon \quad &&\text{ if $p$ is an even-weight polynomial of degree $k$ with $\|p\|_{\infty} \leq 1$}. \tag{P3}
\end{align}
Then $\varepsilon = \varepsilon(n,k,d)$ gives the convergence rate of the hierarchy.  

Note that \eqref{eq:property3} says that $\mathbf K$ should be almost reproducing for such $p$. Therefore we consider the linear operator corresponding to a slight modification of the reproducing kernel:
\[
K = \sum_{r=0}^{2d} c_r C_r^{n},
\]
where $C_r^{n}$ is defined in \eqref{eq:defCrn}. 
In particular, this gives
\begin{equation}
    \label{eq:Kp_cpk}
\mathbf Kp = {\Tr_2(K(x,\sigma)(I_{2^n} \otimes p(\sigma)))} = \sum_{r=0}^{2d} c_r\mathbf C_r^{n}p = \sum_{r=0}^{2d} c_r p_r,
\end{equation}
{
where $\mathbf C_r^{n}$ is the linear operator corresponding to $C_r^n$, and $p_r = \mathbf C_r^{n}p = \Tr_2(C_r^n(x,\sigma)(I_{2^n} \otimes p(\sigma)))$ contains all terms of $p$ with degree equal to $r$, since $C_r^n$ is the reproducing kernel for homogeneous polynomials of degree $r$.}
In particular, if all terms in $p$ are of even degree, then $p_{2r+1} = 0$ for every $r \geq 0$.

To satisfy property \eqref{eq:property1}, we need 
\[
c_0 = \mathbf K1  = 1.
\]

Define 
\[
\widehat C_r^{n} = \frac{C_r^{n}}{\|C_r^{n}\|^2}.
\]
Then by orthogonality of the kernels, we have
\[
\langle \widehat C_r^{n}, K\rangle = c_r.
\]

\subsection{Property \eqref{eq:property2}}
\label{sec:property2}
Suppose that the coefficients $c_r$ are such that
\[
f(i) = \sum_{r=0}^{2d} c_r (-1)^r \mathcal K_r^{n,4}(i)
\]
is a univariate sum of squares of degree $2d$.
Using equation~\eqref{eq:K_to_krawtchouk}, we then have that for any $p$ with $p(\sigma) \succeq 0$, 
\[
\mathbf Kp = \Tr_2(K(x,\sigma)(I_{2^n} \otimes p(\sigma))) = \Tr_2\left(I_{2^n} \otimes p(\sigma)^{\frac12} f\left(\sum_j A_j(x,\sigma)\right) I_{2^n} \otimes p(\sigma)^{\frac12}\right) 
\]
where 
\[
B^{\frac12} = \sum_i \sqrt{\lambda_i} v_iv_i^\dagger
\]
for a positive semidefinite matrix $B = \sum_i \lambda_i v_iv_i^\dagger$, and $A^\dagger$ denotes the conjugate transpose of $A$.
Given an orthonormal basis $\{u_i\}_i$ of $\C^{2^n}$, the partial trace can be evaluated as 
\[
\mathbf Kp
=
\sum_i
(I_{2^n}\otimes u_i^\dagger)
(I_{2^n}\otimes p(\sigma)^{\frac 12})
f\Big(\sum_j A_j(x,\sigma)\Big)
(I_{2^n}\otimes p(\sigma)^{\frac 12})
(I_{2^n}\otimes u_i),
\]
which is a sum of squares of degree $2d$ in the variables $x$.

\subsection{Property \eqref{eq:property3}}
By equation \eqref{eq:Kp_cpk}, we have that for a degree $k$ polynomial of even weight
\[
\mathbf K^{-1}p - p = \sum_{r=0}^{k/2} \left(\frac{1}{c_{2r}} -1\right) p_{2r}.
\]
The extremal eigenvalues of $\mathbf K^{-1}p-p$ are then bounded in absolute value by
\[
\sum_{r=0}^{k/2} \left| \frac{1}{c_{2r}} -1 \right| \max_{0 \leq r \leq k/2} \|p_{2r}\|_{\infty}.
\]

\subsubsection{Bounding $\max_r \|p_{2r}\|_\infty$}
\begin{lemma}\label{lem:component_bound}
    Suppose $H = \sum_{r=0}^k H_r$ is a $k$-local  Hamiltonian on $n$ qubits, with $H_r$ containing terms of weight exactly $r$. Then 
    \[
    \|H_{r}\|_{\infty} \leq \gamma_k\|H\|_{\infty},
    \]
    where the constant $\gamma_k$ is independent of $n$. In particular, if $p$ is a polynomial of degree $k$ in the Pauli matrices,
    \[
    \|p_{r}\|_{\infty} \leq \gamma_k\|p\|_{\infty}.
    \]
\end{lemma}
\begin{proof}
For $\alpha \in [0,1]$, define $H(\alpha) = \sum_{r=0}^k \alpha^r H_r$. The map $H \mapsto H(\alpha)$ is also known as the noise operator, and is completely positive for $\alpha \in [-1/3, 1]$ \cite[Section~8]{cj10-01}. Furthermore, as mentioned in \cite{cj10-01}, $\|H(\alpha)\|_{\infty} \leq \|H\|_{\infty}$ for $\alpha \in [0,1]$ because $H(\alpha)$ can be written as a convex combination of conjugations of $H$ by unitary matrices.

Suppose $0 \leq \alpha_0 < \dots < \alpha_k \leq 1$. Then we can express $H_r$ in terms of $H(\alpha_i)$ by
\[
H_r = \sum_{i=0}^k b_{r,i} H(\alpha_i),
\]
where the coefficients $b_{r,i}$ can be found by inverting the Vandermonde matrix $(\alpha_i^r)_{i,r=0}^k$. 

This implies the bound
\[
\|H_r\|_{\infty} \leq \sum_{i=0}^k |b_{r,i}| \|H(\alpha_i)\|_{\infty} \leq \sum_{i=0}^k |b_{r,i}| \|H\|_{\infty}.
\]
Note that, in particular, the coefficients $b_{r,i}$ are independent of the dimension $n$, so we may take 
$\gamma_k = \max_{0\le r\le k}\sum_{i=0}^k |b_{r,i}|$. 

For a polynomial $p$ whose variables correspond to Pauli matrices, the spectral norm is defined by
\[
\|p\|_\infty = \|p(\sigma)\|_{\infty}.
\]
Because $p_r(\sigma) = {\Tr_2(C_r^n(\sigma,\sigma)(I_{2^n} \otimes p(\sigma)))}= (p(\sigma))_r$, this implies
\[
\|p_r\|_\infty \leq \gamma_k \|p\|_\infty. \qedhere
\]
\end{proof}
In Appendix~\ref{app:bound_gamma_k}, we show the upper bound
\[
\gamma_k < (1+\sqrt{2})^{2k+1}
\]
by choosing a specific set of points $\{\alpha_i\}_i$.

\subsubsection{Bounding $\sum_{r=1}^{k/2}|1/c_{2r} - 1|$}\label{sec:bound_sum_c2r}
We first bound the sum by a quantity that is linear in $c_{2r}$.
{
Recall that to satisfy Property~\ref{eq:property2}, we require 
$f = \sum_{r=0}^{2d} c_r (-1)^r \mathcal K_r^{n,4}$ to be a sum-of-squares polynomial, thus nonnegative. This implies that all eigenvalues of $K$ are nonnegative. 
Combining this with Lemma~\ref{lem:nk1} and Lemma~\ref{lem:eigCrn}, we obtain}
\[
c_{2r} = \langle \widehat C_{2r}^{n}, K\rangle = \langle (-1)^{2r}\widehat{\mathcal K}_{2r}^{n,4}, \sum_j c_j (-1)^{j}\mathcal K_j^{n,4} \rangle_4 \leq \langle   \widehat{\mathcal K}_0^{n,4}, \sum_j c_j (-1)^j\mathcal K_j^{n,4}\rangle_4 = c_0 = 1. 
\]
Now if $c_{2r} \geq \frac{1}{\eta}$ {for some $\eta > 0$}, then 
\[
\sum_{r=1}^{k/2}\left|\frac{1}{c_{2r}} - 1\right| = \sum_{r=1}^{k/2}\frac{1}{c_{2r}} (1-c_{2r}) \leq \eta\sum_{r=1}^{k/2} (1-c_{2r}).
\]

The following lemma is stated for general $q$, but we will only use the $q=4$ case.

\begin{lemma}\label{lem:bound_coeffs}
There is a univariate sum-of-squares polynomial $f = \sum_{r=0}^{2d} c_{r} (-1)^r \mathcal K_r^{n,q}$ such that $c_0 = 1$ and
\[
\sum_{r=1}^{k/2} (1-c_{2r}) \leq \frac{qk(k+2)}{4(q-1)n} \xi_{d+1}^{n,q}.
\]
\end{lemma}
\begin{proof}
{Let $\Sigma_{d}$ denote the cone of univariate sums-of-squares polynomials of degree at most $2d$.}
To find a bound, we want to solve
\begin{equation}
    \label{eq:sos_h}
\begin{aligned}
    & \min  && \langle h, f\rangle_{q} \\
    & \text{ s.t.} && \langle 1, f\rangle_q = 1 \\
    & && f \in \Sigma_{d}
\end{aligned}
\end{equation}
with $h = \frac{k}{2} - \sum_{r=1}^{k/2} \widehat{\mathcal K}_{2r}^{n,q}$. Since $\widehat{\mathcal K}_{2r}^{n,q}(i) \leq \widehat{\mathcal K}_{2r}^{n,q}(0) = 1$, $h$ is minimal at $0$ with $h(0) = 0$. 
Furthermore, if $f = \sum_{r=0}^{2d} c_{r} (-1)^r \mathcal K_{r}^{n,q} \in \Sigma_d$ with $c_0 = \langle 1, f\rangle_q = 1$, then $\langle h, f\rangle_q = \sum_{r=1}^{k/2} (1-c_{2r})$.

By Lemma~\ref{lem:aff_bound_k}, the function $h$ is upper bounded as
\[
h(i) \leq \sum_{r=1}^{k/2} \frac{2qr}{(q-1)n}i = \frac{qk(k+2)}{4(q-1)n} i.
\]
By \cite[Theorem~12]{slot_sum--squares_2023} (which combines results from \cite{de_klerk_worst-case_2020}), the optimal solution to \eqref{eq:sos_h} is attained for the linear function $h(i) = i$ and has optimal value equal to the smallest root $\xi_{d+1}^{n,q}$ of the Krawtchouk polynomial $\mathcal K_{d+1}^{n,q}$. This gives that there is an {$f \in \Sigma_d$ as above} whose coefficients satisfy 
\[
\sum_{r=1}^{k/2} |c_{2r} -1 | = \sum_{r=1}^{k/2}(1-c_{2r}) \leq \frac{qk(k+2)}{4(q-1)n}\xi_{d+1}^{n,q}.\qedhere
\]
\end{proof}
Note that by Lemma~\ref{lem:aff_bound_k}, we also have
\begin{equation}
    \label{eq:bound_c2r}
1-c_{2r} = \langle 1- \widehat{\mathcal K}_{2r}^{n,4}, f\rangle_4 \leq \frac{8r}{3n}\langle i, f\rangle_4 = \frac{8r}{3n}\xi_{d+1}^{n,4},
\end{equation}
for the polynomial $f \in \Sigma_d$ used in the proof of Lemma~\ref{lem:bound_coeffs}. Therefore the condition holds if
\[
1- \frac{8r}{3n}\xi_{d+1}^{n,4} \geq \frac{1}{\eta}
\]
for all $r$. Since the left-hand side is minimal for $r=k/2$, and a larger $\eta$ worsens the bound, the optimal $\eta$ is given by
\[
\eta = \left(1-\frac{4k}{3n}\xi_{d+1}^{n,4}\right)^{-1}.
\]
Thus if $1-\frac{4k}{3n}\xi_{d+1}^{n,4} > 0$, then we have the bound
\[
\sum_{r=1}^{k/2}\left|\frac{1}{c_{2r}}-1\right| \leq \eta \frac{k(k+2)}{3n}\xi_{d+1}^{n,
4}.
\]
For simplicity, Theorem~\ref{thm:main} uses $\eta = 2$ with the corresponding condition on $d$.

\begin{proof}[Proof of Theorem~\ref{thm:main}]
Choose the coefficients $c_r$ from Lemma~\ref{lem:bound_coeffs} applied with $q=4$. Then
$\mathbf K1=1$, and property \eqref{eq:property2} holds by the construction above.
By \eqref{eq:bound_c2r}, the assumption
\[
\frac{4k}{3}\frac{\xi_{d+1}^{n,4}}{n}\le \frac12
\]
implies $c_{2r}\ge 1/2$ for every $1\le r\le k/2$. Therefore
\[
\sum_{r=1}^{k/2}
\left|\frac1{c_{2r}}-1\right|
\le
2\sum_{r=1}^{k/2}(1-c_{2r})
\le
\frac{2k(k+2)}{3}
\frac{\xi_{d+1}^{n,4}}{n}.
\]
Using Lemma~\ref{lem:component_bound}, and the assumption
$\|p\|_\infty\le1$, we get
\[
\|\mathbf K^{-1}p-p\|_\infty
\le
\frac{2k(k+2)}{3}\gamma_k
\frac{\xi_{d+1}^{n,4}}{n}.
\]
The bound on $\gamma_k$ from Appendix~\ref{app:bound_gamma_k} gives the
stated constant $C(k)$. Lemma~\ref{lem:properties_K} then implies
\[
\lambda_{\min}(p)-\nu_d(p)
\le
C(k)\frac{\xi_{d+1}^{n,4}}{n}. \qedhere
\]
\end{proof}

\begin{remark}\label{rmk:finite_conv}
    When $d=n$, the problem \eqref{eq:sos_h} has optimal value $0$ for $h(i) = i$, which  recovers finite convergence at level $n$. To see this, take $f = g^2$, where we choose $g$ such that $g(i) = 0$ for $i \in \{1, \dots, n\}$ and $g(0) = 1$. Then $g$ is of degree $n$, and since $\langle 1, f\rangle_q = g(0) = 1$, it is a feasible solution with $\langle h, f\rangle_q = 0$.
\end{remark}

\begin{remark}
The proofs in \cite{fang_sum--squares_2021, slot_sum--squares_2023, slot_sum--squares_2024} use that
\[
 1-c_{r} \leq \sum_r |1-c_r| 
\]
to get a bound on the minimum degree required for the convergence rates for commutative polynomial optimization on the sphere, the binary cube, and the ball and the simplex, respectively. 
Using a bound on $c_r$ similar to \eqref{eq:bound_c2r} would improve the minimum level for which the convergence rates hold. For example, for the binary cube this improves the requirement on $d$ from $k(k+1)\xi_{d+1}^{n,2}/n \leq 1/2$ to $2k\xi_{d+1}^{n,2}/n \leq 1/2$, and for the ball and the simplex from $d \geq Cnk\sqrt{k}$ to $d \geq Cnk$, where in each case $n$ is the number of variables, $d$ the level of the hierarchy, and $k$ the degree of the polynomial to be minimized.
\end{remark}

\section{Proof of Theorem \ref{thm:second}}\label{sec:proof_thm2}
The dual of the hierarchy in \cite{klep2025upperboundhierarchiesnoncommutative} is given by  
\[
\begin{aligned}
  \mu_d(p)  & = \inf && \langle p, s\rangle, \\
    &\phantom{{}={}}\text{s.t.} && \langle 1, s\rangle = 1, \\
    &&& s \in \Sigma_{d}\langle x\rangle/\mathcal I.
\end{aligned}
\]
To give a convergence rate for this hierarchy, we use a similar idea as in \cite{slot_sum--squares_2024}, adapted to the noncommutative setting. We still require that $p$ is a noncommutative polynomial in the Pauli matrices with monomials of even weight, of even degree $k$.

Let $K$ be the kernel with coefficients $c_{r}$ given in Lemma~\ref{lem:bound_coeffs}, and let $\mathbf K$ denote the
associated linear operator on the Pauli polynomial space.
Let $v$ be a normalized eigenvector of $p(\sigma)$ corresponding to the smallest eigenvalue $\lambda_{\min}(p)$, and take \[
s = (v \otimes I_{2^n})^* K(\sigma, x) (v \otimes I_{2^n}),
\] 
i.e., we evaluate the first factor of the kernel on the Pauli matrices, while the second factor remains a polynomial in $x$. 
Since the univariate polynomial
associated to $K$ is a sum of squares, the same argument as in
Section~\ref{sec:property2} shows that
\[
s\in \Sigma_d\langle x\rangle/\mathcal I .
\]
Moreover, because $K$ reproduces constants,
\[
\langle 1,s\rangle
=
v^\dagger(\mathbf K1)(\sigma)v
=
v^\dagger v
=
1.
\]
Thus $s$ is feasible for the upper-bound hierarchy.

Using the defining property of the kernel,
\[
\langle p,s\rangle
=
v^\dagger(\mathbf Kp)(\sigma)v.
\]
Therefore
\[
\begin{aligned}
\langle p,s\rangle-\lambda_{\min}(p)
&=
v^\dagger(\mathbf Kp-p)(\sigma)v  \\
&\le
\|\mathbf Kp-p\|_\infty  \\
&\le
\sum_{r=1}^{k/2}
|c_{2r}-1|\,\max_{1\le r\le k/2}\|p_{2r}\|_\infty .
\end{aligned}
\]
By Lemma~\ref{lem:bound_coeffs} and Lemma~\ref{lem:component_bound},
\[
\langle p,s\rangle-\lambda_{\min}(p)
\le
\frac{k(k+2)}{3}\gamma_k
\frac{\xi_{d+1}^{n,4}}{n}
=
\frac{C(k)}{2}
\frac{\xi_{d+1}^{n,4}}{n}.
\]
Since $\mu_d(p)\le \langle p,s\rangle$, this proves
\[
\mu_d(p)-\lambda_{\min}(p)
\le
\frac{C(k)}{2}\frac{\xi_{d+1}^{n,4}}{n}.
\]

Unlike the lower-bound estimate in Theorem~\ref{thm:main}, this argument does
not require a lower bound on the coefficients $c_{2r}$, because no inverse
coefficients $1/c_{2r}$ appear. Hence the estimate holds for every relaxation
level $d$.

\section{Discussion and outlook}\label{sec:concl}
In this work, we provided the first convergence rates for hierarchies for noncommutative polynomial optimization when the variables represent Pauli matrices. We showed that for $k$-local Hamiltonians where each term acts on an even number of qubits, both the upper bound and the lower bound hierarchy have error at most\looseness=-1
\[
C(k) \frac{\xi_{d+1}^{n,4}}{n}
\]
in the $d$-th level of the hierarchy, where $C(k)$ does not depend on the number  $n$ of qubits or the level $d$ of the hierarchy.

When the polynomial $p$ only involves the variables $x_{\sigma_Z,i}$, the corresponding Hamiltonian is diagonal in the computational basis. In this case, spectral optimization reduces to binary polynomial optimization on $\{\pm1\}^n$. Indeed, evaluating $p$ on the matrices $\sigma_Z^j$ gives a diagonal matrix whose diagonal entries are $p(z)$, for $z\in\{\pm1\}^n$, and the trace inner product becomes 
\[\langle p_1,p_2\rangle = \frac{1}{2^n}\sum_{z\in\{\pm1\}^n}p_1(z)^*p_2(z). \] 
Thus the Pauli setting considered here contains the binary cube as a commutative special case.

A natural question is whether the same construction extends to qudit systems. A direct analogue of the kernels used here appears to lead to 
\[(1-d')^r\mathcal K_r^{n,q}\!\left(\sum_j A_j^{d'}\right), \qquad q=\frac{2d'}{d'-1}, \] 
for local dimension $d'$, where $A_j^{d'}$ is the qudit analog of the matrices $A_j = \frac{3}{4} I + \frac{1}{4}\sum_{W \in \{X, Y, Z\}} \sigma_W^j \otimes \sigma_W^j$. For kernels of the form $\sum_r c_r C_r^{n,d'}$ with nonnegative eigenvalues, this forces bounds of the form 
\[c_{2r}\le (d'-1)^{-2r}c_0. \] 
Thus, unlike in the qubit case, this direct construction cannot produce sum-of-squares kernels whose even coefficients $c_{2r}$ remain close to $1$. Extending the convergence-rate analysis to qudits therefore seems to require a different kernel construction.

Another interesting research direction is to exploit additional algebraic
structure. For example, for the quantum max-cut problem, the hierarchy can be highly simplified when one uses the swap matrices
\[
S_{ij} = \frac{1}{2}I_{2^n} + \frac{1}{2}\sum_{W \in \{X, Y, Z\}} \sigma_W^i \sigma_W^j.
\]
This then leads to a noncommutative polynomial optimization problem using a different ideal \cite{watts2024relaxations}. 
Although this symmetry-reduced hierarchy converges at level $\lceil n/2\rceil$ rather than level $n$, finite convergence alone does not imply better quantitative bounds at low levels. It would be interesting to obtain convergence rates that reflect this additional algebraic structure.

Sparsity is another important direction. Correlative sparsity can substantially reduce the size of the semidefinite programs~\cite{klep_sparse_2022}. At a fixed relaxation level, a sparse hierarchy may give weaker bounds than the dense hierarchy, but the reduced computational cost may allow one to reach higher levels. In the commutative setting, convergence rates for sparse sum-of-squares hierarchies were obtained in~\cite{korda2025convergence}. A natural next step is to develop analogous rates for sparse noncommutative hierarchies and compare the resulting accuracy as a function of computational cost.

Finally, the finite-dimensional nature of the Pauli setting is essential to our argument. The variables are represented by explicit matrices, which gives a canonical trace inner product and concrete reproducing kernels. This contrasts sharply with dimension-free noncommutative polynomial optimization, where one optimizes over Hilbert spaces of arbitrary dimension, operators on those spaces, and states. 
{Our use of reproducing kernels is related to the
noncommutative Christoffel--Darboux kernels of~\cite{ncchristoffel};
developing this connection further may provide a systematic way to construct
kernels for other structured noncommutative optimization problems.}
In full generality, effective convergence rates cannot be expected. Indeed, for Bell inequalities and nonlocal games, the identities $\mathrm{MIP}^{co}=\mathrm{coRE}$ \cite{Lin2025MIPco} and $\mathrm{MIP}^*=\mathrm{RE}$ \cite{MIP*} rule out uniform computable convergence rates for the corresponding SDP hierarchies: such rates would yield algorithms for undecidable approximation problems. 
Thus future convergence-rate results must rely on additional structure, such as fixed finite-dimensional representations, explicit matrix constraints, symmetry, sparsity, or other assumptions that make the feasible set quantitatively controllable.
Without such restrictions,
uniform computable convergence rates are ruled out by undecidability phenomena.\looseness=-1

\section*{Acknowledgments}
This work has been supported by European Union’s HORIZON-MSCA-2023-DN-JD programme under the Horizon Europe (HORIZON) Marie Skłodowska-Curie Actions, grant agreement 101120296 (TENORS), the project COMPUTE, funded within the QuantERA II Programme that has received funding from the EU's H2020 research and innovation programme under the GA No 101017733 \euflag. 
IK also acknowledges support of the Slovenian Research Agency program P1-0222 and grants J1-50002, N1-0217, J1-60011, J1-50001, J1-3004 and J1-60025. 
Partially supported by the Fondation de l’\'Ecole polytechnique
as part of the Gaspard Monge Visiting Professor Program. IK  thanks \'Ecole polytechnique and Inria Paris Saclay
for hospitality during the preparation of this manuscript.

\section*{Data availability}
No data were created or analyzed in this study.

 \appendix
 \section{Reduction to even-weight Hamiltonians}
\label{app:evenweight}
\begin{lemma}\label{lem:evenweight}
Let
\[
H= H_{\mathrm{odd}} + H_{\mathrm{even}} %cI+H_1+H_2
\]
be a $k$-local Hamiltonian on $n$ qubits, where
$H_{\mathrm{odd}}$ contains the terms with odd weight and 
$H_{\mathrm{even}}$ contains the terms with even weight.
Define a Hamiltonian on $n+1$ qubits by
\[
\widetilde H:= H_{\mathrm{odd}} \otimes \sigma_Z^0 + H_{\mathrm{even}}\otimes I_2,
\]
where $\sigma_Z^0$ acts on the ancilla qubit.

Then $\widetilde H$ is $k'$-local, where
$k'=k$ if $k$ is even, and $k'=k+1$ if $k$ is odd. Moreover, every monomial in $\widetilde H$ has even degree, and
\[
\spec(\widetilde H)=\spec(H)
\]
as sets of eigenvalues. In particular,
\[
\lambda_{\max}(\widetilde H)=\lambda_{\max}(H),
\qquad
\lambda_{\min}(\widetilde H)=\lambda_{\min}(H).
\]
\end{lemma}

\begin{proof}
By construction, the terms of $H_{\mathrm{even}}\otimes I_2$ are still of even degree, and every  term of odd degree
$\prod_{i \in J} \sigma^i_{\alpha_i}$ in $H_{\mathrm{odd}}$ is replaced by the even degree term $\prod_{i \in J}\sigma^i_{\alpha_i} \sigma_Z^0$.
Hence every monomial in $\widetilde H$ has even degree.

Since $\sigma_Z^0$ has eigenvalues $\pm1$, in the $\sigma_Z^0$-eigenbasis we have
\[
\widetilde H \cong (H_{\mathrm{even}} + H_{\mathrm{odd}}) \oplus (H_{\mathrm{even}} - H_{\mathrm{odd}}),
\]
{where the symbol ``$\cong$'' represents unitary equivalence.}
Thus it suffices to show that $H_{\mathrm{even}}-H_{\mathrm{odd}}$ and $H_{\mathrm{even}}+H_{\mathrm{odd}}$ have the same spectrum.

Let $\kappa$ denote entrywise complex conjugation on $\mathbb C^2$, and define
the antilinear operator
\[
J:=\sigma_Y\kappa
\]
on $\C^2$. 
Since
\[
\overline {\sigma_X}=\sigma_X,\qquad \overline {\sigma_Y}=-\sigma_Y,\qquad \overline {\sigma_Z}=\sigma_Z,
\]
a direct computation gives
\[
J\sigma_XJ^{-1}=\sigma_Y\overline {\sigma_X}\,\sigma_Y=-\sigma_X,\qquad
J\sigma_YJ^{-1}=\sigma_Y\overline {\sigma_Y}\,\sigma_Y=-\sigma_Y,\qquad
J\sigma_ZJ^{-1}=\sigma_Y\overline {\sigma_Z}\,\sigma_Y=-\sigma_Z.
\]
Now let
\[
\Theta:=J^{\otimes n},
\]
an antilinear isometry on $(\mathbb C^2)^{\otimes n}$.
If $w$ is a Pauli word of degree $m$, then
\[
\Theta\, w\, \Theta^{-1}=(-1)^m w.
\]
In particular, every odd-degree Pauli monomial changes sign, while every
even-degree Pauli monomial is fixed. 
Since $H$ is Hermitian, its coefficients in the Pauli-word basis are real. Therefore
\[
\Theta(H_{\mathrm{even}}+H_{\mathrm{odd}})\Theta^{-1}=H_{\mathrm{even}}-H_{\mathrm{odd}}.
\]

Since $\Theta$ is bijective, it preserves the spectrum:
if $(H_{\mathrm{even}}+H_{\mathrm{odd}})v=\lambda v$, then
\[
(H_{\mathrm{even}}-H_{\mathrm{odd}})\Theta v
=
\Theta(H_{\mathrm{even}}+H_{\mathrm{odd}})v
=
\lambda\,\Theta v.
\]
Hence
\[
\spec(H_{\mathrm{even}}-H_{\mathrm{odd}})=\spec(H_{\mathrm{even}}+H_{\mathrm{odd}})=\spec(H).
\]
Combining this with the block decomposition above, we obtain
\[
\spec(\widetilde H)=\spec(H),
\]
as claimed.
\end{proof}

\section{Proofs of Lemma~\ref{lem:nk1} and Lemma \ref{lem:aff_bound_k}}\label{app:proofs_lemmas}
For the proofs of the lemmas, we will use  the generating function of the Krawtchouk polynomials, which is given by
\[
\sum_{r=0}^n \mathcal K_r^{n,q}(i) z^r = (1-z)^i (1+(q-1)z)^{n-i}.
\] 
\begin{lemma}[Restatement of Lemma \ref{lem:nk1}]
For every $0 \le i \le n$ and $0 \le r \le n$,
\[
\mathcal K_r^{n,q}(i) \le \mathcal K_r^{n,q}(0)=(q-1)^r\binom{n}{r}.
\]
\end{lemma}

\begin{proof}
We compare coefficients. First,
\[
(1-z)^i \ll (1+z)^i \ll (1+(q-1)z)^i,
\]
where $\ll$ denotes coefficientwise inequality. Hence
\[
(1-z)^i(1+(q-1)z)^{n-i}
\ll
(1+z)^i(1+(q-1)z)^{n-i}
\ll
(1+(q-1)z)^n.
\]
Taking the coefficient of $z^r$ {(denoted by $[z^r]$)} in the generating function gives
\[
\mathcal K_r^{n,q}(i)
=
[z^r]\bigl((1-z)^i(1+(q-1)z)^{n-i}\bigr)
\le
[z^r](1+(q-1)z)^n
=
(q-1)^r\binom{n}{r}.
\]
Since
\[
\mathcal K_r^{n,q}(0)=[z^r](1+(q-1)z)^n=(q-1)^r\binom{n}{r},
\]
the claim follows.
\end{proof}

\begin{lemma}[Restatement of Lemma~\ref{lem:aff_bound_k}]
For every $0 \le i \le n$ and $1 \le r \le n$,
\[
\mathcal K_r^{n,q}(0)-\mathcal K_r^{n,q}(i)
\le
q \cdot (q-1)^{r-1}\binom{n-1}{r-1}\, i.
\]
Equivalently, 
\[
1-\widehat{\mathcal K}_r^{n,q}(i)\le \frac{qr}{(q-1)n}\,i.
\]
\end{lemma}

\begin{proof}
Using again the generating function,
\[
\begin{split}
\mathcal K_r^{n,q}(0)-\mathcal K_r^{n,q}(i)
& =
[z^r]\Bigl((1+(q-1)z)^n-(1-z)^i(1+(q-1)z)^{n-i}\Bigr) \\
& 
=
[z^r]\Bigl((1+(q-1)z)^{n-i}\bigl((1+(q-1)z)^i-(1-z)^i\bigr)\Bigr).
\end{split}
\]
Then
\[
\begin{split}
(1+(q-1)z)^i-(1-z)^i
& =
\bigl((1+(q-1)z)-(1-z)\bigr)
\sum_{s=0}^{i-1}(1+(q-1)z)^{i-1-s}(1-z)^s
\\
& =
qz\sum_{s=0}^{i-1}(1+(q-1)z)^{i-1-s}(1-z)^s.
\end{split}
\]
Since $(1-z)^s \ll (1+(q-1)z)^s$ coefficientwise, we get
\[
(1+(q-1)z)^i-(1-z)^i
\ll
qiz(1+(q-1)z)^{i-1}.
\]
Therefore
\[
\mathcal K_r^{n,q}(0)-\mathcal K_r^{n,q}(i)
\le
[z^r]\Bigl(qiz(1+(q-1)z)^{n-1}\Bigr)
=
qi\,[z^{r-1}](1+(q-1)z)^{n-1}=qi \cdot (q-1)^{r-1}\binom{n-1}{r-1},
\]
proving the first claim.

For the normalized version, simply divide by
\[
\mathcal K_r^{n,q}(0)=(q-1)^r\binom{n}{r}. \qedhere
\]
\end{proof}

\section{Upper bound on $\gamma_k$}\label{app:bound_gamma_k}
To find an explicit bound on $\gamma_k$, we choose the interpolation points $\alpha_i$ to be the Chebyshev nodes shifted to the interval $[0,1]$, which are the roots of the shifted Chebyshev polynomial of the first kind $T^*_{k+1}(x) = T_{k+1}(2x-1)$. 
Here the superscript  $*$ in $T_j^*$ is part of the standard notation for the
shifted Chebyshev polynomial and is unrelated to the involution on
$\C\langle x\rangle$ used elsewhere in the paper.

These polynomials satisfy the discrete orthogonality relation 
\[
\sum_{i=0}^{k} T_j^*(\alpha_i)T_l^*(\alpha_i) = \frac{k+1}{u_j}\delta_{jl}
\]
for $j,l \leq k$, where $u_j = 1$ if $j = 0$ and $u_j = 2$ otherwise.

Recall that we need to express $H_r$ from the expression
\[
H(\alpha)= \sum_r H_r\alpha^r
\]
in terms of the evaluations on $\alpha_i$. That is, we require the coefficients $b_{r,i}$ such that
\[
H_r = \sum_{i=0}^k b_{r,i} H(\alpha_i),
\]
which then give the bound $\gamma_k \leq \max_{r} \sum_{i=0}^k |b_{r,i}|$.

The discrete orthogonality gives the expression
\[
H(\alpha) = \sum_{j=0}^k \frac{u_j}{k+1}(\sum_{i=0}^{k} H(\alpha_i) T_j^*(\alpha_i) )T_j^*(\alpha), 
\]
and hence
\[
H(\alpha) = \sum_{j=0}^k \frac{u_j}{k+1}\sum_{i=0}^k H(\alpha_i) T_j^*(\alpha_i) \sum_{r=0}^j \tau_{r, j} \alpha^r
\]
where $\tau_{r,j}$ is the coefficient of the degree $r$ term of $T_j^*$.
Thus the coefficients $b_{r,i}$ are given by
\[
b_{r,i} = \sum_{j=r}^k\frac{u_j}{k+1}T^*_j(\alpha_i) \tau_{r,j}.
\]
This gives the bound
\[
\gamma_k \leq \sum_{r,i} |b_{r,i}| \leq 2\sum_{j=0}^k \sum_{r=0}^j |\tau_{r,j}|.
\]

Because all roots of the shifted Chebyshev polynomial $T_j^*(x)$ are positive, its coefficients must alternate in sign. Furthermore, the leading coefficient of $T_j(2x-1)$ is strictly positive for $j \ge 1$. This implies that $\text{sgn}(\tau_{r,j}) = (-1)^{j-r}$.

This sign alternation allows us to write the sum of the absolute values of the coefficients as an evaluation of the polynomial itself:
\[
    \sum_{r=0}^j |\tau_{r,j}| = \sum_{r=0}^j \tau_{r,j} (-1)^{j-r} = (-1)^j \sum_{r=0}^j \tau_{r,j} (-1)^r = (-1)^j T_j^*(-1).
\]

Using the parity property of the standard Chebyshev polynomials, $T_j(-x) = (-1)^j T_j(x)$, we simplify this expression to:
\[
    (-1)^j T_j^*(-1) = (-1)^j T_j(-3) = T_j(3).
\]

Substituting this back into our initial bound for $\gamma_k$ yields the summation:
\[
  \gamma_k \le 2\sum_{j=0}^k T_j(3).
\]

For $x \ge 1$, the Chebyshev polynomials of the first kind can be evaluated using the closed-form expression $T_j(x) = \frac{1}{2}\left((x + \sqrt{x^2-1})^j + (x - \sqrt{x^2-1})^j\right)$. Thus
\[
    T_j(3) = \frac{1}{2}\left((3 + 2\sqrt{2})^j + (3 - 2\sqrt{2})^j\right) < (3 + 2\sqrt{2})^j.
\]

We can now cleanly bound the sum by evaluating the geometric series of the dominant term:
\begin{align*}
    \sum_{j=0}^k T_j(3) &< \sum_{j=0}^k (3 + 2\sqrt{2})^j \\
    &= \frac{(3 + 2\sqrt{2})^{k+1} - 1}{(3 + 2\sqrt{2}) - 1}  \\
    &< \frac{(3 + 2\sqrt{2})^{k+1}}{2(1+\sqrt{2})} \\
    & =  \frac{1}{2} \left(1+\sqrt{2}\right) \left(3+2
   \sqrt{2}\right)^k
\end{align*}

Therefore, we obtain the upper bound:
\[
    \gamma_k <  \left(1+\sqrt{2}\right) \left(3+2
   \sqrt{2}\right)^k = \left(1+\sqrt{2}\right)^{2k+1}.
\]

%\bibliography{refs}
%\bibliographystyle{unsrt}

\end{document}